\documentclass[copyright,creativecommons]{eptcs}
\providecommand{\event}{GaM 2016} 
\usepackage{breakurl}             
\usepackage{graphicx}
\usepackage{subfigure}
\usepackage{xspace}
\usepackage{listings}
\usepackage{paralist}
\usepackage[all]{xy}
\usepackage{wrapfig}
\usepackage{amsmath}
\usepackage{amsthm}

\newtheorem{theorem}{Theorem}
\newtheorem*{proof*}{Proof}
\newtheorem{lemma}{Lemma}

\title{Type Annotation for Adaptive Systems}
\author{Paolo Bottoni
\institute{Sapienza University, Rome, Italy}
\institute{Dep. of Comp. Sc.}
\email{bottoni@di.uniroma1.it}
\and
Andrew Fish \institute{ University of Brighton, Brighton, UK}
\institute{Sch. of Comp., Eng. and Math.}
\email{Andrew.Fish@brighton.ac.uk}
\and
Francesco Parisi Presicce
\institute{Sapienza University, Rome, Italy}
\institute{Dep. of Comp. Sc.}
\email{parisi@di.uniroma1.it}
}
\def\titlerunning{Type Annotation for Adaptive Systems}
\def\authorrunning{P. Bottoni, A. Fish \& F. Parisi Presicce}

\begin{document}

\maketitle

\begin{abstract}
We introduce type annotations as a flexible typing mechanism for graph systems and discuss their advantages with respect to classical typing based on graph morphisms. In this approach the type system is incorporated with the graph and elements can adapt to changes in context by changing their type annotations. We discuss some case studies in which this mechanism is relevant.
\end{abstract}

\lstset{ %
  basicstyle=\footnotesize,        
  breakatwhitespace=false,         
  breaklines=true,                 
  escapeinside={\%*}{*)},          
  frame=single,	                   
  keepspaces=true,                 
  keywordstyle=\bfseries,       
  language=OCL,                 
  otherkeywords={let,in,context,inv,forall,exists},      
  tabsize=2,	                   
}

\section{Introduction}\label{sec:intro}
Typing systems define the admissible structures and behaviours of models formed by their instances, possibly defining relations among types in terms of inheritance and composition.
Such a general notion is at the basis both of object-oriented analysis and design, and of metamodeling techniques, for both of which the relations with the field of graph transformations have been explored~\cite{DBLP:journals/entcs/FerreiraR05,DBLP:journals/sosym/EhrigKT09}.

When using an approach based on graph transformation theory, admissible structures are defined by graph constraints, either explicitly presented as such or included in the graph type, while behaviours are defined by collections of rules, possibly regulated by some control mechanism.
This supports a view of an application domain as defined by a type graph and a policy for model generation (i.e. a collection of constraints) and evolution (i.e. a collection of rules).
However, the development of concrete applications often requires bringing together concerns that are expressed with respect to different domains, or are cross-domains.
In such cases, two possibilities are typically offered: either to merge different domains in a comprehensive one, which might become inflexible in view of possible evolutions of the application, or to maintain the domains separated, but to build explicit connections among them through specific constructs.
In particular, triple graphs are a way to construct explicit relations between types in two different domains, allowing the definition of traces of model transformations~\cite{SK08}.
In a series of recent papers we have proposed annotations to add context-specific information in a flexible way  to elements of a graph, while preserving the separation of the original
domains~\cite{DBLP:journals/eceasst/BottoniP13,DBLP:journals/vlc/BottoniP13}.

In this paper, we bring the notion of annotation to bear on the context of model adaptation, by proposing to replace the notion of typing as traditionally expressed via a typing morphism from an instance graph to a type graph with the notion of annotation of model elements with type information, thus embedding the type information in the model itself.
This achieves three important objectives:

\begin{enumerate}
  \item Instances can be typed in a dynamic way, by updating the type annotation associated with them (see the example in Section~\ref{sec:sex}). Moreover, in contrast to what is required in metamodel evolution, individual instances of one same type may follow different evolutions, i.e. some change to a type, some change to another, and some remain of the same type, adapting to their specific contexts. Moreover, type graphs may be enriched, rather than substituted.
  \item Multiple typing is inherently supported, unless otherwise constrained, since typing information is maintained via specific graph patterns and not via typing morphisms.
  \item Types do not univocally define the structure of their instances, but express some conformance requirement (see the example in Section~\ref{sec:classification}).
%
%
\end{enumerate}

From a philosophical point of view, we might say that a type, rather than defining an ontological concept, defines a set of observations that we can perform on an element and through which we can deem it as an  instance of that type, while not excluding that a different set of observations would lead to a different categorisation of the same individual element under a different type. This is particularly relevant when considering the use of some object as a resource for some process other than the ones it was originally intended. For example, to a producer a cardboard box is a \emph{packaging material} defined by number of measures, from grammage to edge crush resistance. Once it has entered a house and has been emptied of its content, its usefulness as a \emph{container} for other material can be established based on its width, length and height, while to evaluate its safety as a \emph{support} for a toddler who wants to climb on it, a new type of measure, what could be called ``resistance to toppling'', would be needed.

This paper starts this line of investigation, enriching the notion of annotation with that of type annotation.
Firstly, in Section~\ref{sec:related}, related work is briefly discussed, followed by condensed formal background notions in Section~\ref{sec:background}.
Next, Section~\ref{sec:typing} deals with the process of transforming typing into annotations, including  inheritance.
%
%
Then, Section~\ref{sec:dynamic} provides the model of dynamic typing in the context of annotation types and discusses the impact this
has with respect to constraints associated with a type, whilst Section~\ref{sec:caseStudies} provides extracts of case studies showing how the use of type annotation can model some typical situations. We conclude in Section~\ref{sec:concls}. 

\section{Related Work}\label{sec:related}
In a series of papers we have proposed annotations as a flexible way to integrate contextual information into application domains and have discussed ways to derive application conditions and repair actions from the presence of annotations~\cite{DBLP:journals/eceasst/BottoniP13,DBLP:journals/vlc/BottoniP13}.

In~\cite{LGC15}, the authors decouple the two aspects of typing: as a blueprint for creation and as a way of classifying elements. They propose to rely on standard mechanisms for object creation and use \emph{a-posteriori} typing at the type level to relate two different metamodels, and at the instance level as a means to reclassify objects and enable multiple, partial, dynamic typings. Our proposal achieves objectives similar to instance-level a-posteriori typing, also dealing with changes of type within the same hierarchy.

In~\cite{DBLP:journals/tois/GottlobSR96}, class hierarchies are complemented by role hierarchies, whose nodes represent role types that an object classified in the root may take on. At any point in time, an entity is represented by an instance of the root and an instance of every role type whose role it currently plays. Annotations can be used to express type and role information simultaneously.

The notion of annotation and its application to typing offers possibilities for the extension of graph transformation approaches to the Semantic Web, exploiting explicit relations between instances and concepts, analogous to the formalisation of RDF proposed in~\cite{DBLP:journals/eceasst/BraatzB08}.

A survey on the notion of adaptation has been presented in~\cite{BCGLV12}, mainly with reference to programming, and proposing an initial formalisation of the dynamics of adaptation based on Labeled Transition Systems, without a specific accent on typing. 

\section{Background}\label{sec:background}
A \emph{category} is a construct
$\mathbf{C}=({Ob}(\mathbf{C}),{Hom}(\mathbf{C}),\circ)$ where
${Ob}(\mathbf{C})$ is a collection of objects, ${Hom}(\mathbf{C})$ is a collection of (homo)morphisms $m:a\rightarrow{b}$ for some
$a,b\in{Ob}(\mathbf{C})$ and
$\circ:{Hom}(\mathbf{C})\times{Hom}(\mathbf{C})\rightarrow{Hom}(\mathbf{C})$ is the composition operation such that:
\begin{enumerate}
\item $\forall{o}\in{Ob}(\mathbf{C})$ $\exists{id_o}\in{Hom}(\mathbf{C})$, with ${id_o}\circ{m}={m}\circ{id_o}={m}$ for any ${m}\in{Hom}(\mathbf{C})$;
\item  $\forall{m_1},m_2,m_3\in{Hom}(\mathbf{C})$ with $m_1:a\rightarrow{b}$, $m_2:b\rightarrow{c}$, $m_3:c\rightarrow{d}$,
$(m_3\circ{m_2})\circ{m_1}=m_3\circ({m_2}\circ{m_1})$.
\end{enumerate}

We set our treatment in the category $\mathbf{Graph}$, defined by graphs and graph morphisms~\cite{DBLP:series/eatcs/EhrigEPT06}.
A (directed) \emph{graph} is  a tuple $(V,E,s,t)$, with $V$ and $E$
finite sets of \emph{nodes} and \emph{edges}, and functions
$s:E\rightarrow V$,
$t:E\rightarrow{V}$ mapping an edge to its source and target. The notion is extended to that of (directed)
\emph{graph with boxes}~\cite{DBLP:journals/eceasst/BottoniP13} as a tuple $G=(V,E,B,s,t,cnt)$, where:
\begin{inparaenum}[(1)]
\item $V$ and $E$ are sets of nodes and edges as in usual graphs;
\item $B$ is a set of boxes, such that $B\cap(V\cup{E})=\emptyset$;
\item the \emph{source} and \emph{target} functions $s$ and $t$ extend their
    codomains to $V\cup{B}$;
\item $cnt:B\rightarrow\wp(V\cup{B})$ is a function associating a box with
    its
    \emph{content}\footnote{Here and elsewhere $\wp$ denotes the powerset.} with the properties that $b\not\in{cnt(b)}$ and if $x\in{cnt(b_1)}$ and $b_1\in{cnt(b_2)}$, then $x\in{cnt(b_2)}$.
\end{inparaenum}

We refer to graphs with boxes as \emph{B-graphs}
or just graphs unless it is necessary to distinguish them.

A \emph{type B-graph} is a distinguished graph
$TG$ $=$ $(V_T,E_T,B_T,s^T,t^T,cnt^T)$, where $V_T$, $E_T$
and $B_T$ are sets of node, edge and box types, respectively, while the
functions $s^T:E_T\rightarrow{V_T}\cup{B_T}$ and
$t^T:E_T\rightarrow{V_T\cup{B_T}}$ define source and target node- and box-
types for each edge type, and the function
$cnt^T:B_T\rightarrow\wp(V_T\cup{B_T})$ associates each type of box with the
set of types of elements it can contain.

A \emph{morphism}  $f:G_1\rightarrow{G_2}$ between \emph{B-graphs}, with
${G_i}=(V_i,E_i,B_i,s_i,t_i,cnt_i)$ for $i=1,2$, is a triple
$(f_V:V_1\rightarrow{V_2},f_E:E_1\rightarrow{E_2},f_{B}:B_1\rightarrow{B_2})$
that preserves source, target and inclusion functions, i.e.,
$f_{V\cup{B}}\circ{s_1}={s_2}\circ{f_E}$,
$f_{V\cup{B}}\circ{t_1}=t_2\circ{f_E}$, and if $x\in{cnt_1(b)}$ for some $b\in{B_1}$ and $x\in{V_1\cup{B_1}}$, then
$f_{V\cup{B}}(x)\in{cnt_2}(f_{B}(b))$, where $f_{V\cup{B}}$ is defined as the union of $f_V$ and $f_{B}$.

A B-graph $G$ is \emph{typed on} a type B-graph $TG$ if there is a graph
morphism $\mathit{tp}^G:G\rightarrow{TG}$, with $\mathit{tp}^G_V:V\rightarrow{V_T}$,
$\mathit{tp}^G_E:{E}\rightarrow{E_T}$ and $\mathit{tp}^G_{B}:B\rightarrow{B_T}$ s.t.
$\mathit{tp}^G_V(s(e))$ $=$ $s^T(\mathit{tp}^G_E(e))$ and $\mathit{tp}^G_{V\cup{B}}(t(e))$ $=$ $t^T(\mathit{tp}^G_E(e))$ such that
$\forall{b}\in{B}\forall{x}\in{cnt(b)}[\mathit{tp}^G_X(x)\in{cnt^T}(\mathit{tp}^G_{B}(b))]$,
where $X$ is one of ${V,B}$, depending on the type of $x$.
A morphism $f:G_1\rightarrow{G_2}$ between $TG$-typed graphs preserves the type,
i.e. $\mathit{tp}^{G_{2}}\circ{f}=\mathit{tp}^{G_{1}}$.
%
%
A \emph{graph transformation rule} is a span of graph morphisms
\smash{$L\stackrel{l}{\leftarrow}{K}\stackrel{r}{\rightarrow}{R}$}, and is
applied following the Double Pushout Approach.
A \emph{graph constraint} is a
morphism $c:P\rightarrow{C}$, \emph{satisfied} by a graph $G$ if for each
morphism $m_p:P\rightarrow{G}$, a \emph{match} morphism
$m_c:C\rightarrow{G}$ exists, with
$m_c\circ{c}=m_P$. A \emph{negative graph constraint} (or \emph{forbidden graph}) is defined as $(\neg{c}):P\rightarrow{P}$ and is satisfied by $G$ only if no such match for $P$ is found in $G$.
An \emph{application condition} for
\smash{$L\stackrel{l}{\leftarrow}{K}\stackrel{r}{\rightarrow}{R}$} is a morphism
$ac:L\rightarrow{AC}$, such that the rule is applicable on a match
$m_L:L\rightarrow{G}$ if a match $m_{ac}:AC\rightarrow{G}$ exists such that
$m_{ac}\circ{ac}=m_L$.
A \emph{negative application condition} requires such an
$m_{ac}$ not to exist.

Annotations of elements
of a domain ${\cal{D}}_1$
with nodes of a domain ${\cal{D}}_2$ are defined via nodes
from \texttt{AnnotationNode}. We call $\cal{A}$ the domain of such
annotation nodes. A graph is \emph{well-formed with respect to annotations}, if each 
%
%
node 
$a\in{\cal{A}}$ participates in exactly one instance of the 
\emph{annotation pattern}
$\pi_a=x\stackrel{e_1}{\leftarrow}{a}\stackrel{e_2}\rightarrow{y}$,
where $x\in{\cal{D}}_1$, $y\in{\cal{D}}_2$, $e_1$ is an edge of type \texttt{annotates}
and $e_2$ is an edge of
type \texttt{with}.

While we use attributes in the definition of the metamodel presented in Section~\ref{sec:typing}, in this paper we work with graphs that are just typed, without considering attributes as part of their structure. For all practical purposes, we can model attribution as a special kind of annotation of an element with some value in some domain, and we can express through suitable constraints the fact that elements of a given type must be endowed with some set of attributes, each taking values of some given type.

\section{From typing to annotations}\label{sec:typing}
We
%
%
set
our work in the collection of models conforming to the metamodel $\mathbf{M}$ shown as a UML model in
Figure~\ref{fig:metamodel}. A \emph{domain} is defined by a type \emph{graph}
(composed of \emph{model elements})
and a \emph{policy}, and policies are defined by a collection of \emph{morphisms}. Except for \texttt{Domain}, which is defined by a collection of policies referring to some common type graph,
all other notions derive from that of \texttt{Element}, which possibly has a name, so that all elements may be involved in annotations.
As specified in Section~\ref{sec:background}, the \texttt{src} and \texttt{tgt} associations, representing the $s$ and $t$ functions,
can only relate instances of \texttt{Node} (but not of \texttt{AnnotationNode}, as specified by a constraint not shown here) or \texttt{Box} to an \texttt{Edge}.

\begin{figure}[htb]
\centering
  \includegraphics[width=15.1cm]{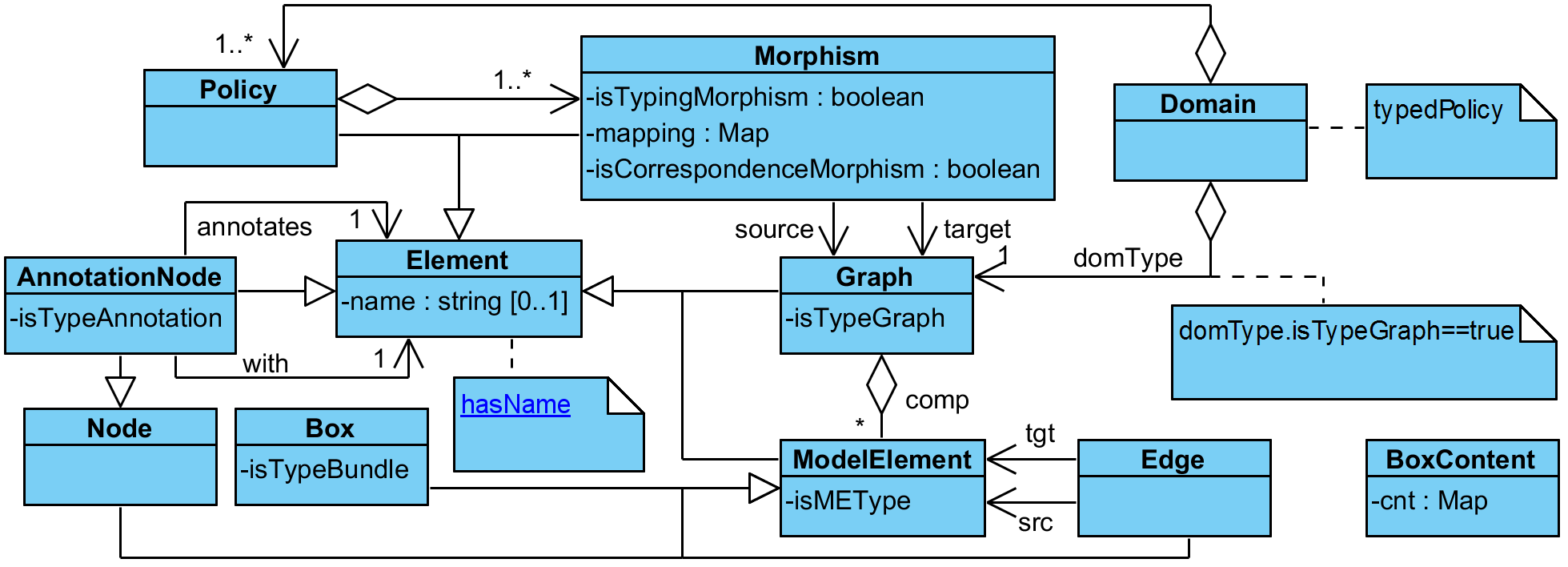}
  \caption{The metamodel $\mathbf{M}$ for incorporating types into graphs, also giving access to morphism maps.}
  \label{fig:metamodel}
\end{figure}

Types are defined, as per the value of a boolean attribute, as special
model
elements, and
%
%
analogously for
typing morphisms. All type elements have a name.
Any element can be the argument (i.e. be reached by an edge of type \texttt{annotates}) or the value (reached by an edge of type \texttt{with}) of some annotation, respecting the constraint for well-formedness, with suitable restrictions described below for type annotations. We consider here  morphisms specifying rules or (positive) constraints.
The following global constraints, expressed in OCL, require that all morphisms in the policy for a domain are typed on the type graph associated with the domain:

\begin{lstlisting}
context Domain inv
let
	allMor : Set = Morphism.allInstances(),
	type : Graph = self.domType
in

type.isTypeGraph = true and

self.policy.morphism -> forall(m | allMor ->
  exists(m2,m3 |
	m2.isTypingMorphism = true and m2.source = m.source and m2.target = type and
	m3.isTypingMorphism = true and m3.source = m.target and m3.target = type))
\end{lstlisting}

%
Morphisms are here considered as first class elements, whose actual specification (i.e. the function relating elements in the two graphs composing the morphism) resides in a map.
Similarly, we reify the $cnt$ function for the box content and specify it through a map.
Moreover, some local constraints (not shown here) 
define the type annotation pattern, extending the constraint of well-formedness
%
%
discussed in Section~\ref{sec:background}, whereby each annotation node must annotate exactly one element with exactly one element.

The \texttt{annNodeInstance} and \texttt{annNodeType} constraints in Figure~\ref{fig:annTypeToNodeType} ensure that node instances are annotated with type nodes and that type nodes are used only to annotate node instances. Analogous constraints are defined for edges and boxes so that type annotations are consistent with the sorts of the involved elements. Forbidden graphs prevent type annotation nodes to be used for elements of any other sort. For an element $x$ of the \texttt{Node}, \texttt{Edge} or \texttt{Box} sort, we will denote by $annType(x)$ the (set of) type(s) with which $x$ is annotated. For simplicity, whenever $annType(x)$ is a singleton we will also use the same notation to indicate its only element. The condition
\texttt{isTypeBundle=true} indicates a particular kind of box which can contain a bundle of types and that will be used in
Sectione~\ref{sec:inheritance} to model inheritance.


\begin{figure}[htb]
\centering
  \subfigure{\includegraphics[height=3.2cm]{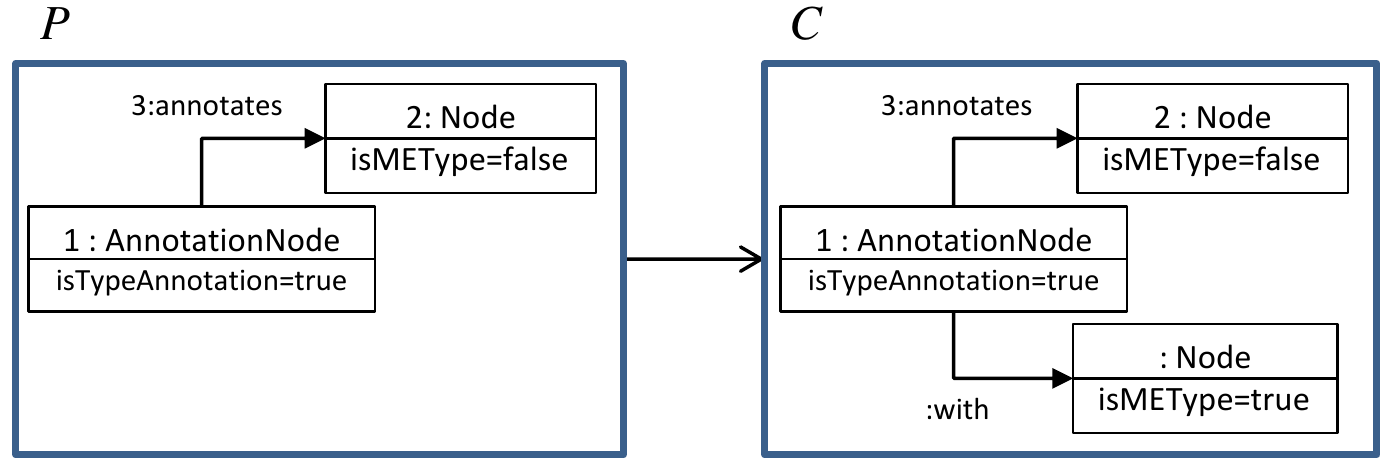}}
  \subfigure{\includegraphics[height=3.2cm]{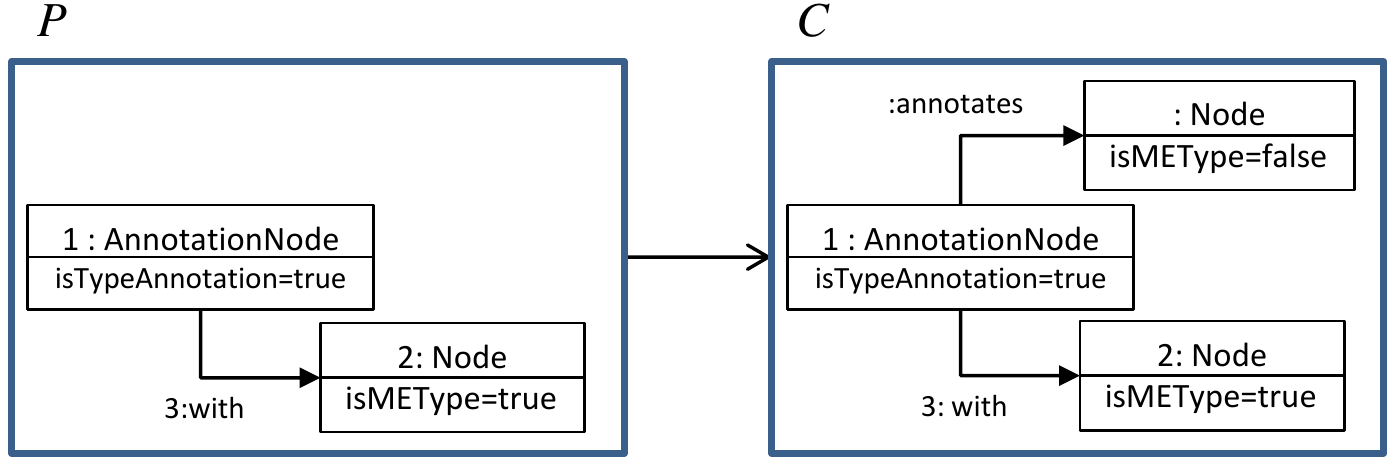}}
  \caption{Constraints \texttt{annNodeInstance} and \texttt{annNodeType} for consistency with the \texttt{Node} sort.}
  \label{fig:annTypeToNodeType}
\end{figure}

The forbidden graph \texttt{notTypedTwice} in Figure~\ref{fig:notTypedTwice} expresses the fact that an element cannot be typed in the same way twice (combined with the constraints in Figure~\ref{fig:annTypeToNodeType} such an element can only be a node, edge or box instance). It does not exclude, however, that an element be typed in two (or more) different ways, as will be discussed in Section~\ref{sec:dynamic}.
A further constraint, not shown here due to lack of space, expresses that if there are two edges annotated with the same edge type,
then their source nodes are annotated with the same node type, and so are their target nodes.
%

\begin{figure}[htb]
\centering
  \includegraphics[height=3.2cm]{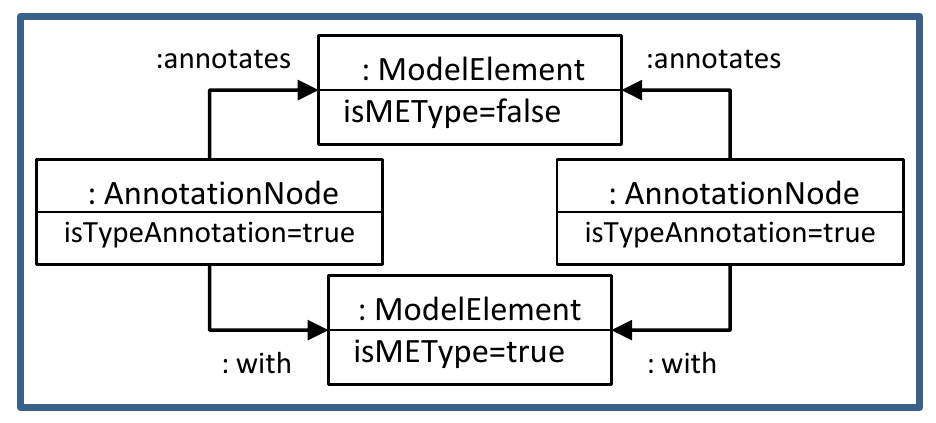}
  \caption{The forbidden graph \texttt{notTypedTwice}: An element cannot be typed in the same way twice.}
  \label{fig:notTypedTwice}
\end{figure}

We say that a graph $G$ is \emph{with type annotation} (or that it is a \emph{type-annotated graph}),
\emph{iff} there exists a typing of $G$ on the metamodel $\mathbf{M}$ and $G$ satisfies all of the constraints discussed above. Note that all these constraints are concerned with well-formedness, while correctness of type annotations within specific domains must be expressed through domain-dependent constraints.

Graphs with type annotations can be used as an alternative to the use of typing morphisms, to associate type information with elements.
In particular, with reference to Figure~\ref{fig:construction}(a), let $G$ be a graph typed on the type graph $TG$ (via the $\mathit{tp}^G$ morphism, so that each element of $G$ has a single type) and let $H(G,\mathit{tp}^G)$ (in the following simply denoted by $H$ where $G$ and $\mathit{tp}^G$ are clear from the context) be the minimal graph with type annotation such that:

\begin{enumerate}
	\item both $G$ and $TG$ have isomorphic (disjoint) immersions in $H$, respectively called $G^\prime$ and $TG^\prime$, defined by the morphisms ${f_g}$ and ${f_t}$;
	\item each element in $G^\prime$ is annotated with exactly one type element in $TG^\prime$;
	\item for each element $x$ in $G$, $annType(f_g(x))=f_t(\mathit{tp}^G(x))$.
\end{enumerate}

$H$ is unique up to type-annotation preserving isomorphisms. Considering the coproduct $G\oplus{TG}$ of $G$ and $TG$ and the universal property of coproducts, there is a unique morphism
$f_g\oplus{f_t}:G\oplus{TG}\rightarrow{H}$ such that the triangles in the diagram of Figure~\ref{fig:construction}(a) commute. Moreover, $H$ preserves by construction the type information provided by the $\mathit{tp}^G$ morphism.


\begin{figure}[htb]
	\centering
\subfigure{
 \makebox{
    \xymatrix{
        G \ar@{^{(}->}[rr]_{in_1} \ar@{^{(}->}[drr]|{f_g} \ar@/^1pc/[rrrr]^{\mathit{tp}^{G}}
				& {}
				&
				G\oplus{TG} \ar[d]|{f_g\oplus{f_t}}
				& {}
				& TG \ar@{_{(}->}[ll]^{in_2} \ar@{_{(}->}[dll]|{f_t} \\
       {} & {} & H(G,tp^G) & {} & {}
    }
 }
}
%
\subfigure{
 \makebox{
    \xymatrix@R=0.5pc{
        TG & {} & {} TypeCorr \ar[ll]_{corrType1} \ar[drr]^{corrType2}& {} & {} \\
        {} & {} & {}  & {} & H(G,tp^G) \\
        G \ar[uu]^{\mathit{tp}^G} & {} & InstCorr  \ar[ll]^{corrInstance1} \ar[urr]_{corrInstance2} & {} & {}
    }
 }
}
\caption{The construction of a type-annotated graph (a), and the triple graph metamodel (b).}
\label{fig:construction}
\end{figure}
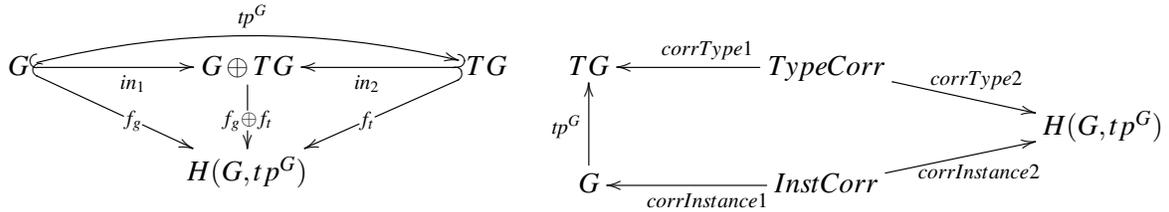

A graph is \emph{correct under type annotation} if all the type annotations for edges are consistent with the restrictions on the type annotations for their sources and targets and all the box contents are consistent with the restrictions on the type annotations for the content of the box type. We say that a morphism is
\emph{type-annotation-preserving} if each element annotated with some type in its source is annotated with the same type in its target.
We call $\mathbf{AT}$ the category of graphs with type annotation
(not attributed) and
%
%
graph morphisms mapping type-annotated graphs into type-annotated graphs, and
$\mathbf{AT}_1$ its restriction to the case where each element is annotated with exactly one single type.
Let $\mathbf{TG}$ be the category of typed graphs (not attributed)
(technically the category of graphs over a type graph) with
%
%
graph morphisms mapping typed graphs into typed graphs.
The construction of Figure~\ref{fig:construction}(a) induces 
%
%
a functor
$\mathtt{typeAnn}:\mathbf{TG}\rightarrow{\mathbf{AT}}_1$ where
$\mathtt{typeAnn}_{Ob}$ maps each object $G$ of $\mathbf{TG}$ into the object
$H(G,\mathit{tp}^G)$ of $\mathbf{AT}_1$ and $\mathtt{typeAnn}_{Hom}$ maps each
%
%
morphism
$m:G\rightarrow{G^\prime}$ of $\mathbf{TG}$ into a morphism
$m^\prime:H(G,\mathit{tp}^G)\rightarrow{H^\prime}(G^\prime,\mathit{tp}^{G^\prime})$ such that for each element $x$ of $G$,
$f_{g^{\prime}}(m(x))=m^\prime(f_g(x))$.
The correctness of the construction is given by Lemmas~\ref{th:lem1} and~\ref{th:lem2}.

\begin{lemma}\label{th:lem1}
If a graph $G\in{Ob(\mathbf{TG})}$ is correct under typing morphisms, then its image under
$\mathtt{typeAnn}_{Ob}$ is correct under type annotation.
\end{lemma}

\begin{proof}[Sketch] The proof is immediate from the construction, considering the constraints for well-formedness discussed above.
\end{proof}


\begin{lemma}\label{th:lem2}
If a morphism $m\in{Hom(\mathbf{TG})}$ is type preserving then
$\mathtt{typeAnn}_{Hom}(m)$ is type-annotation-preserving.
\end{lemma}

\begin{proof}
Let $m:G\rightarrow{G^\prime}$ be a morphism in $\mathbf{TG}$ and for each element $x\in{G}$, let $\mathit{tp}^G(x)$ be its type, with
$\mathit{tp}^G(x)=\mathit{tp}^{G^\prime}(m(x))$ by definition of type-preserving morphism. Let
$\mathtt{typeAnn}_{Hom}(m)=m^\prime:H\rightarrow{H^\prime}$,
with $H=\mathtt{typeAnn}_{Ob}(G)$ and
$H^\prime=\mathtt{typeAnn}_{Ob}(G^\prime)$. Since
$m^\prime(f_g(x))=f_{g^\prime}(m(x))$
%
%
and we have $\mathit{tp}^G(x)=annType(f_g(x))$ for any element $x$ in $G$
(or $\mathit{tp}^{G^\prime}(x)=annType(f_{g^\prime}(x))$ for $x$ in $G^\prime$), we conclude that
$annType(m^\prime(f_g(x)))=annType(f_{g^\prime}(m(x)))=\mathit{tp}^{G^\prime}(m(x))=\mathit{tp}^G(x)=annType(f_g(x))$.
\end{proof}

Since all of the morphisms in the diagram of Figure~\ref{fig:construction}(a)
%
%
preserve all elements from $G$ and $TG$ and $H$ presents only additional information, the original graphs can be recovered from within $H$ and the original typing can be reconstructed from the annotation information. Conversely, for each type-annotated graph $H^\prime$, in which each element is type-annotated with only one type, it is possible to obtain a typed graph $H$ such that its image under the functor
$\mathtt{typeAnn}$ is isomorphic to $H^\prime$.

The discussion above can be extended to the case where $H^\prime$ has elements without any type annotation, or has elements with multiple type annotations. In this case there exists a set $\mathbf{H}$ of typed graphs such that for each $H\in\mathbf{H}$ $\mathtt{typeAnn}_{Ob}(H)$ is only a subgraph of $H^\prime$.

\subsection{Correspondence patterns}\label{sec:tripleGraph}
The construction of Figure~\ref{fig:construction}(a) can also be interpreted in terms of triple patterns~\cite{SK08}
%
%
with reference to the diagram of Figure~\ref{fig:construction}(b). We define a
\emph{composition of triple graphs on a common target}, whose correspondence  graphs,
\emph{TypeCorr} and \emph{InstCorr}, respectively relate the type, $TG$ and instance $G$ graphs, seen as distinct source graphs, to the type-annotated graph, $H$, seen as the common target graph.

A \emph{triple pattern} is a construct $\mathit{TriPatt}=(\mathit{TriP},\Gamma)$,
where $\mathit{TriP}$ is a triple graph and $\Gamma$ is a formula on the elements in $\mathit{TriP}$.
We say that $T\mathit{TriP}$ is \emph{satisfied} by a triple graph $\mathit{TriGrph}$ if each component of $\mathit{TriP}$ has an injective match in the corresponding component of $\mathit{TriGrph}$, preserving the domains and images of the correspondence morphisms and satisfying $\Gamma$. The \emph{composition of triple patterns} is defined similarly to the composition of triple graphs.
A composition of triple patterns is \emph{satisfied by a composition} of triple graphs with a common target if each single triple pattern is satisfied by the corresponding triple graph and all the morphisms in the composition are preserved.

In this setting, Figure~\ref{fig:tripleGraph} presents two compositions of triple patterns relating a typed graph with a  
\begin{wrapfigure}{r}{0.42\textwidth}
\centering
\includegraphics[width=6cm]{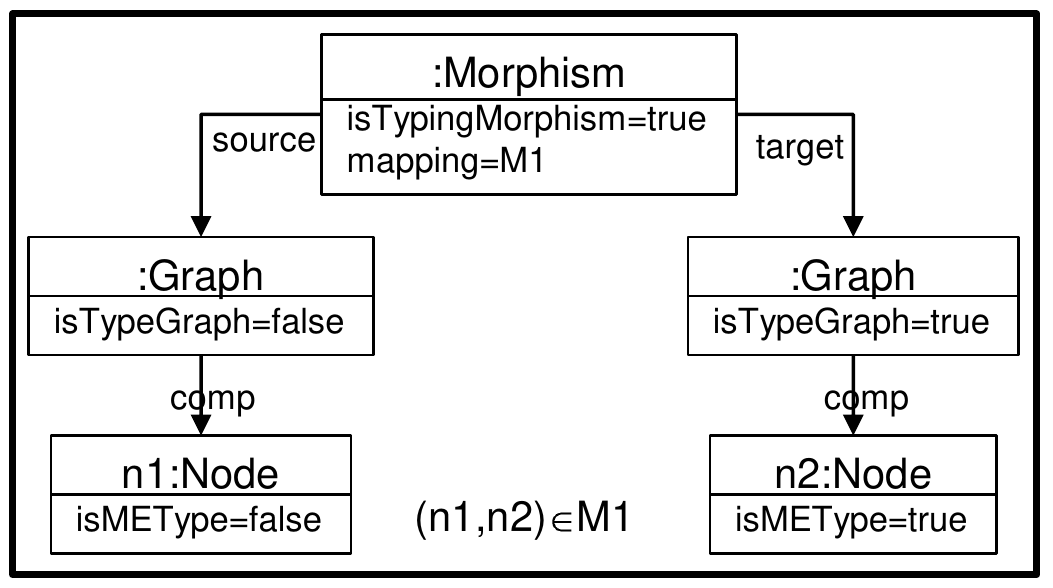}
  \caption{The graph pattern denoted by the \texttt{type} arrow in Figure~\ref{fig:tripleGraph}.}
  \label{fig:tripleGraphPattern}
\end{wrapfigure}
type-annotated graph. All the graphs are typed on the metamodel of
Figure~\ref{fig:metamodel}, the correspondence graph is a discrete graph and the dashed lines relate corresponding elements in some morphism map.
In particular, the \texttt{type} arrow is a shortcut for indicating that the instance and the type elements constitute a pair in the map of the typing morphism between an instance and a type graph, as shown in the composition of triple patterns of
Figure~\ref{fig:tripleGraphPattern} for the case of nodes. Analogous patterns are defined for the different types of graph elements. The relation between the triple patterns and the construction of Section~\ref{sec:typing} is expressed via Theorem~\ref{th:correspondences}. We call $AnnTypePatt$ the set of all these compositions of triple patterns.

\begin{figure}[htb]
\centering
\subfigure{\includegraphics[width=12cm]{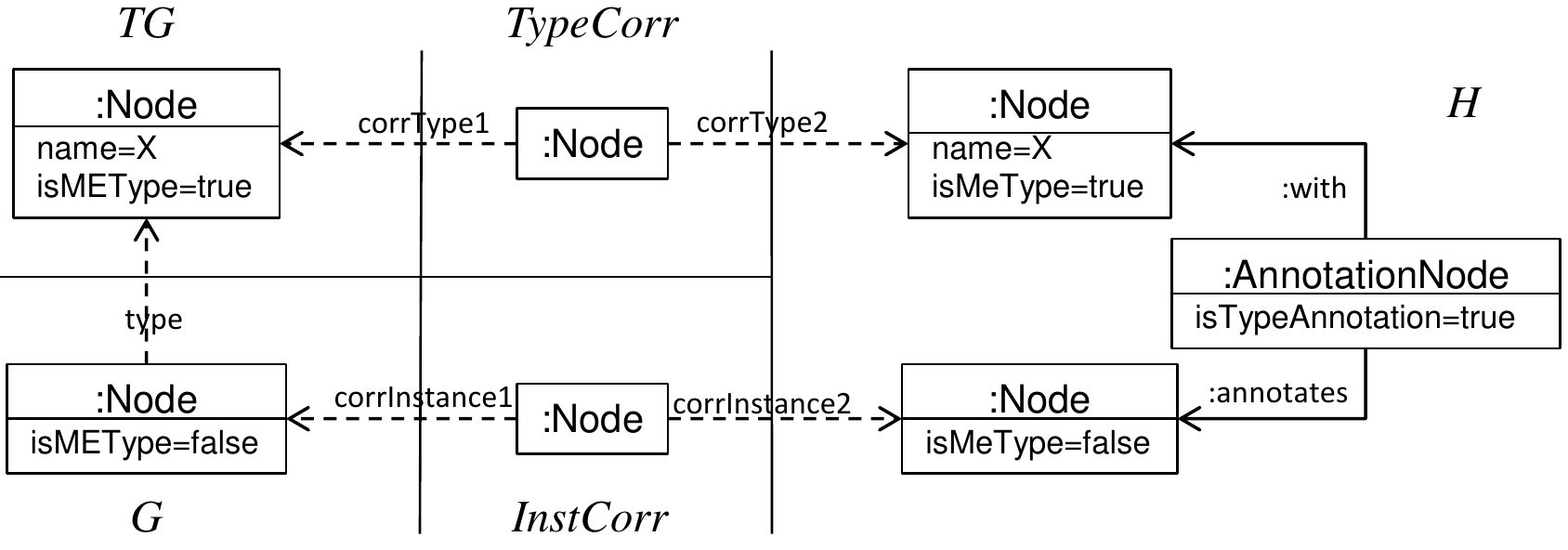}}\\
\rule{12cm}{0.6pt}\\
\subfigure{\includegraphics[width=12cm]{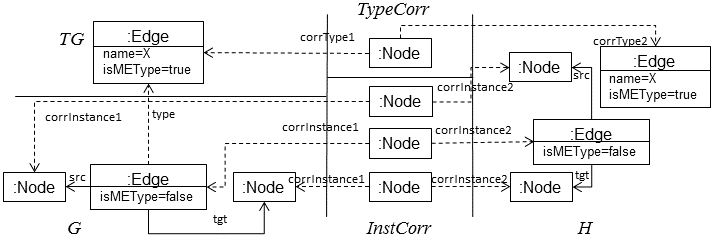}}\\
  \caption{Two compositions of triple patterns describing the relation between typed elements and elements with type annotation for nodes (top) and edges (below).}
  \label{fig:tripleGraph}
\end{figure}

\begin{theorem}\label{th:correspondences}
Let $G$, $TG$ and $H$ be as in Section~\ref{sec:typing}. Then there exists a composition $CMP$ of the triple graphs
$\mathit{TriType}=TG{\xleftarrow{corrType1}}\mathit{TypeCorr}{\xrightarrow{\mathit{corrType2}}}H$,
$\mathit{TriInst}=G{\xleftarrow{\mathit{corrInstance1}}}\mathit{InstCorr}{\xrightarrow{\mathit{corrInstance2}}}H$,
on the common pattern $H$,
such that $CMP$ satisfies all the patterns in $\mathit{AnnTypePatt}$.
\end{theorem}

\begin{proof}
Since $G$ is a graph typed on $TG$, each element in $G$ is associated via $type$ to a unique element in $TG$, according to its sort, so that each element in $G$ provides a unique match for the $G$ component of the source graph in the triple pattern for that sort and the corresponding images of the $\mathit{tp}^G$ morphism provide the matches for the composition of the source graphs. By construction of $H$, all elements in $G$ and $TG$ have a unique copy in $H$, and all copies of elements in $G$ participate in exactly one type annotation, thus providing a unique match for the common target graph. The  graphs $TypeCorr$ and $InstCorr$ are then constructed with a correspondence node for each element in $TG$ and $G$, respectively and the correspondence morphisms relate this node with the copy of the corresponding element in $H$.
\end{proof} 

\subsection{Inheritance}\label{sec:inheritance}
The notion of type annotation can also be extended to consider inheritance.
We assume here that types are organised into a single inheritance hierarchy (whilst permitting the possibility of multiple typing).
Let $(T,\le)$ be a partial order of types, where $T_1\le{T_2}$ indicates that $T_1$ inherits from $T_2$.
We assume the presence of a type $\top$ such that $T_i\le\top$ for each $T_i$.
%
%
If needed, one can organise the set of types into partitions for node, edge and box types.
To preserve the inheritance relation, we add a constraint \texttt{inheritsAnnotation} stating that if an element is annotated with type $T_1$,
it is also annotated with type $T_2$.
Managing type information via annotation also allows the removal of information, but in this case a special process is needed. Let
${T_i}\le\dots\le{T_j}\le\dots\le\top$ be a chain in the partial order,
and let $e$ be an element annotated with $T_i$ (hence with all the types from $T_{i+1}$ to $\top$, including $T_j$).
Removing the annotation with $T_j$ requires also the removal of all of the annotations with types from $T_i$ to $T_j$,
while leaving the annotation with $T_{j+1}$ (hence with all the types from $T_{j+1}$ to $\top$). This can be achieved via the use of boxes for which \texttt{isTypeBundle=true}. In this case, the constraints discussed in Section \ref{sec:typing} are substituted with constraints requiring that all elements be annotated with boxes of this kind, each box containing the types in an inheritance chain. Additional constraints specify that the content of such a box is made of types of the correct sort.
Removing an annotation with type $T_j$ would then amount to substituting an annotation with a box containing the complete chain above with one containing the chain starting from $T_{j+1}$. We assume that annotations with a box containing only $\top$ can never be removed.

\section{Managing dynamic typing}\label{sec:dynamic}
We model dynamic typing in the context of type annotations and discuss its impact with respect to constraints associated with a type. We assume that constraints involving types are of three forms, schematised in
Figure~\ref{fig:dynamicConstraints}, in the form of pattern morphisms,
using a compact notation for pattern representation from~\cite{DBLP:journals/fuin/BottoniGL16}.
A pattern $\pi$ is given by a collection of graphs
$\mathbf{G}=\{G_1,\dots,G_n\}$ and a collection of morphisms
$M_\pi=\{m^{\{i,j\}}_\pi:G_i\rightarrow{G_j}\mid{G_i},{G_j}\in\mathbf{G}\}$ organised into a tree structure rooted in $G_1$ and such that each $G_i\in\mathbf{G}$ is involved in at least one morphism in $M_\pi$.
Each graph is represented as a (named) region, and the tree structure is reproduced by the nesting of regions. We say $\pi$ is satisfied by a graph $G$, denoted $G\models\pi$, if there exists a collection, $S_{\pi,P}$, of morphisms from each $G_i\in\mathbf{G}$ into $G$ preserving the image of each morphism in $M_\pi$. A pattern morphism $m_\Pi:\pi\rightarrow\pi^\prime$ is a collection of injective morphisms from graphs in $\mathbf{G}$ to graphs in $\mathbf{G}^\prime$, preserving the tree structure and the image of each morphism.

Under this notion, the first form (top) requires that an element involved in some specific pattern be annotated with some specific type. The pattern in the premise can possibly include requests on types.
The second form (middle) requires that if a certain pattern of elements of some given types exist, it must be somehow related to an element annotated with some specific type.
The third form (bottom) requires that if an element is typed in a given way, it must be connected with some pattern.
All of these forms can be easily extended to include in the conclusion some additional pattern (not including requests on types). The case in which there are further requests on the presence of typed elements in the conclusion can be dealt with by cascading constraints, so that the conclusion of one constraint becomes the premise of another one, which requires an additional element with some type annotation.

\begin{figure}[htb]
	\centering
\subfigure{\includegraphics[height=2.5cm]{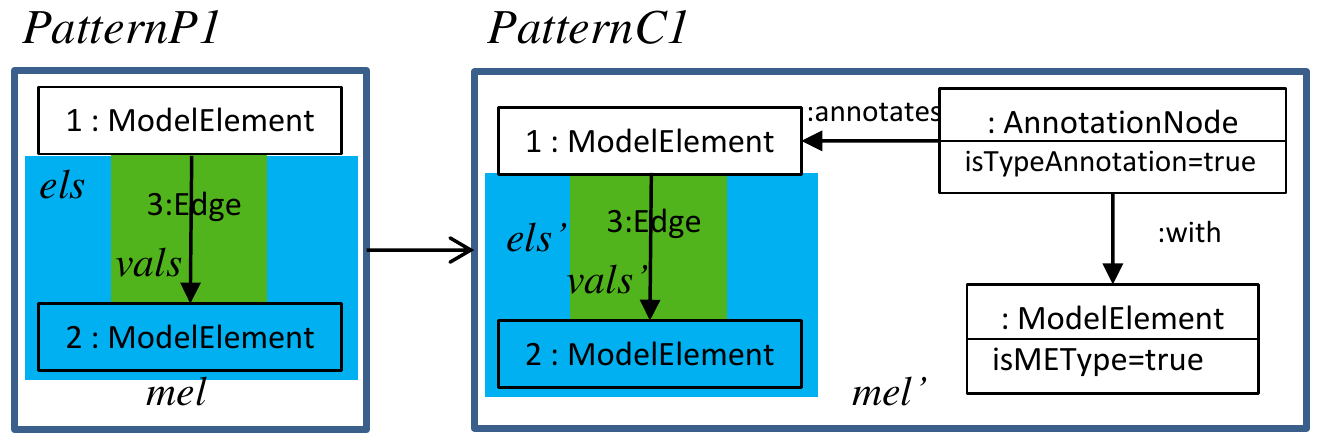}}
\subfigure{\includegraphics[height=2.5cm]{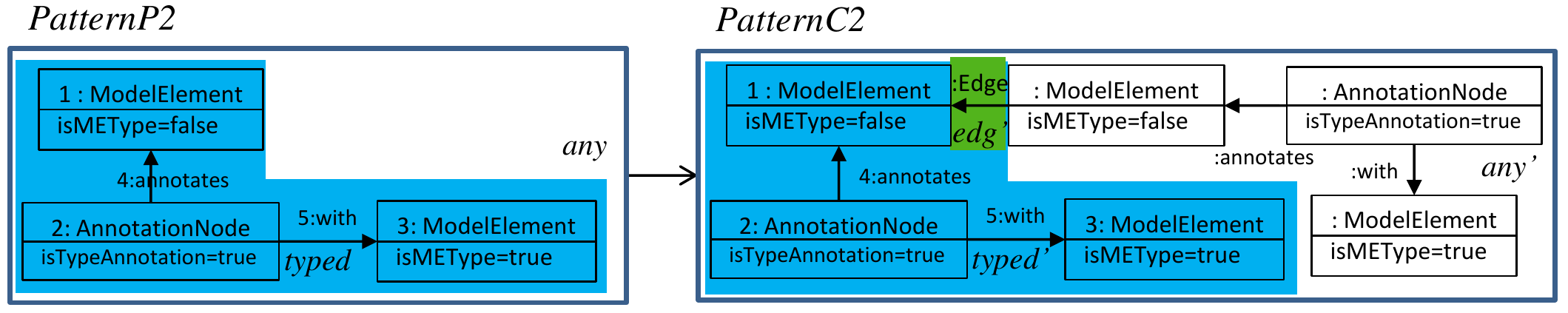}}
\subfigure{\includegraphics[height=2.5cm]{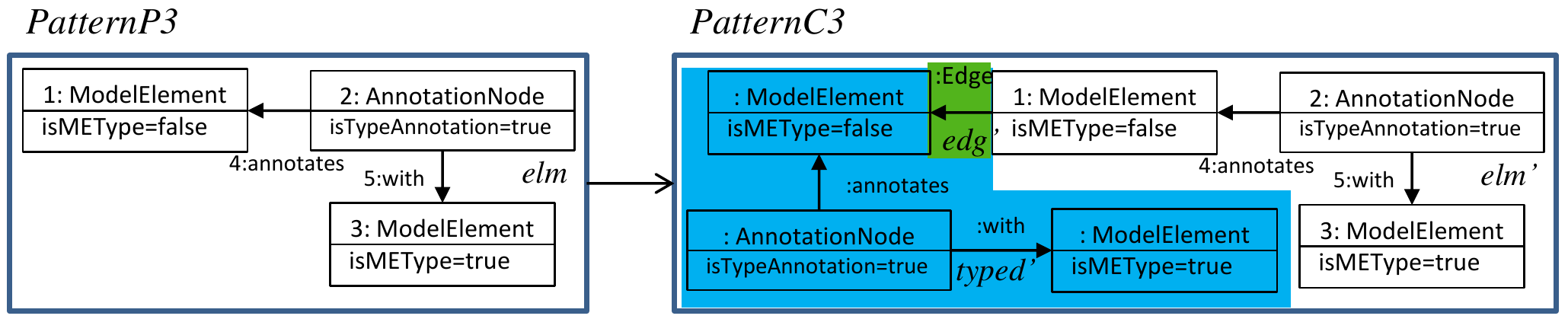}}
\caption{Generic forms of constraints: requiring a specific type for an element involved in a pattern (top); requiring the presence of an element typed in a specific way (middle); requiring the presence of some pattern related to an element typed in a specific way (bottom).}
	\label{fig:dynamicConstraints}
\end{figure}

We also consider that all the rules for changing type annotation are of the form shown in Figure~\ref{fig:changeNodeType}, where only the $L$ and $R$ parts of the rule scheme are given, the $K$ part being formed only by the model elements and the type nodes. Analogous schemes hold for changing the types of edges and boxes.

\begin{figure}[htb]
	\centering
		\includegraphics[width=10cm]{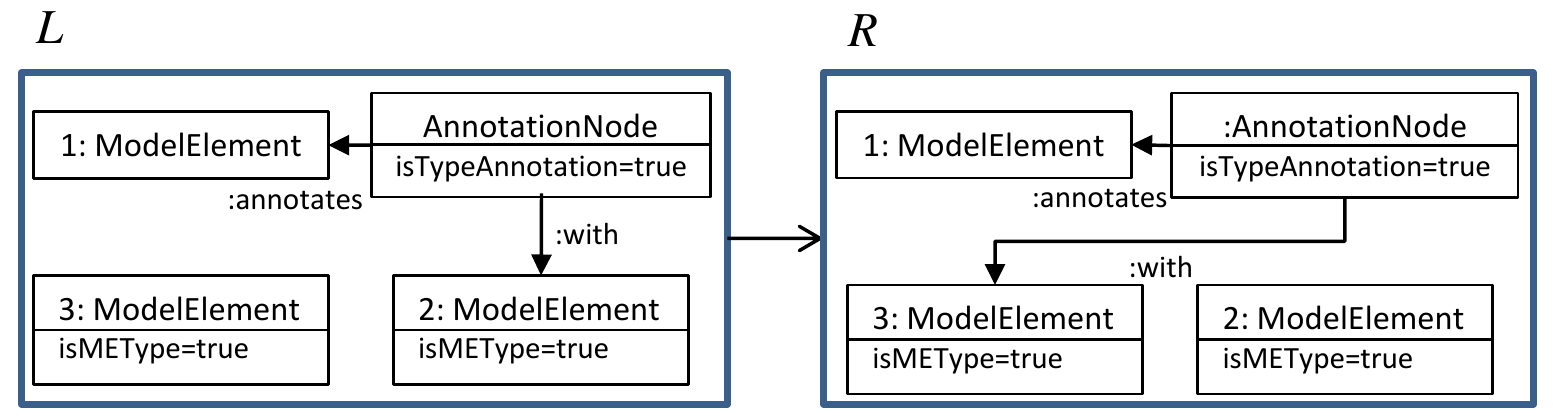}
	\caption{The general form \texttt{changeNodeType} for rules changing a type annotation for a node.}
	\label{fig:changeNodeType}
\end{figure}

Constraints describing the structure of types are not put in peril by a type change and changing the type of an element which was originally part of a pattern in a premise does not lead to any violation of the constraint, since the premise simply ceases to hold. However, constraints requiring that elements with some properties have specific types can cease to hold. In particular, changing the type annotation for an element can have an impact on the well-formedness of a graph in different ways:

\begin{itemize}
  \item Change the type annotation for an element which needed to be typed in some specific way due to a constraint of the first form.
  \item Change the type annotation for the only element whose typing makes a constraint of the second form satisfied.
  \item Change the type annotation of an element in a pattern, thus making the premise of a constraint of either form hold.
\end{itemize}

Therefore, we need a construction to identify the constraints which are violated after a type change
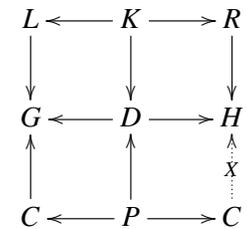
\begin{wrapfigure}{r}{0.3\textwidth}
\vspace{-15pt}
\centering
\makebox{ \xymatrix{
        L \ar[d] & K \ar[l] \ar[r] \ar[d] & R \ar[d] \\
        G & D \ar[l] \ar[r] & H \\
        C \ar[u] & P \ar[l] \ar[u] \ar[r] & C \ar@{.>}[u]|X
    }
}
\caption{Constraint violation due to rule application.}
\label{fig:repair1}
\end{wrapfigure}
occurs. The construction must then identify the repair action that would preserve correctness with respect to type annotation.
We assume that the composition of all constraints is satisfiable by finite graphs,
so that we do not have to consider infinite chains of consequences. However, the problem of knowing if a collection of positive constraints admits finite models is undecidable~\cite{DBLP:journals/fac/OrejasEP10}. A change in some type annotation might lead to a cascade of consequences, a problem over which a partial form of control can be achieved.
\vspace{-0.2pc}
%
Figure~\ref{fig:repair1} shows a situation where a constraint $m:P\rightarrow{C}$, which was previously holding,
is violated after application of a DPO rule $L\leftarrow{K}\rightarrow{R}$ to a graph $G$. Note that, due to the forms of constraints and rules, the match for the premise of the condition is preserved in the transformation.
When dealing with constraints of the type of Figure~\ref{fig:dynamicConstraints} (top) we have two possible solutions to preserve the correctness of the graph with respect to type annotation.
Since we need to disrupt the relation with the premise of the element whose type annotation is going to change,
thus making the premise not valid in the resulting graph, this can be achieved by extending the left-hand side of the rule with the pattern,
or by applying a repair action afterwards.

Figure~\ref{fig:repairs} (left) shows the first solution. The premise is added to the left-hand side, glueing the two graphs in the element $e$ (the one which is connected to the pattern and whose type annotation is going to change), while the restriction $\overline{P}$ of the premise, obtained by removing all the elements connecting $e$ to the pattern is added to the $K$ and $R$ parts, preserving the images of the morphism under the constraint. Figure~\ref{fig:repairs} (right) shows the second solution, consisting of applying a repair action removing the connection of the element with the pattern after applying the rule, matching it to the co-image of $R$ in $H$. The choice between the two constructions is typically domain-dependent.

\begin{figure}[htb]
\centering
\subfigure{\makebox{ \xymatrix{
        L\oplus_e{P} \ar[d] & K\oplus_e{\overline{P}} \ar[l] \ar[r] \ar[d] & R\oplus_e{\overline{P}} \ar[d] & \\
        G & D \ar[l] \ar[r] & H \\
        C \ar[u] & P \ar[ul] \ar@{.>}[u]|X \ar[l]
    }
    }
}
\subfigure{\makebox{ \xymatrix{
        L
        \ar[d]
        & K
        \ar[l] \ar[r] \ar[d]
        & R  \ar[dr]
        & &
       {P} \ar[dl]
  & \overline{P} \ar[l] \ar[r] \ar[d]
  & \overline{P} \ar[d]
  \\
        G & D
        \ar[l] \ar[rr]
        & & H & & D^\prime \ar[ll] \ar[r]
        & H^\prime \\
        C
        \ar[u]
        & P \ar[l] \ar[u] \ar[r]
        & C
        \ar@{.>}[ur]|X
        & & C
        \ar@{.>}[ul]|X
        & P
        \ar@{.>}[u]|X \ar[l]
}}}
\caption{Removing a match for the premise: during rule application (left); via a repair action (right).}
\label{fig:repairs}
\end{figure}
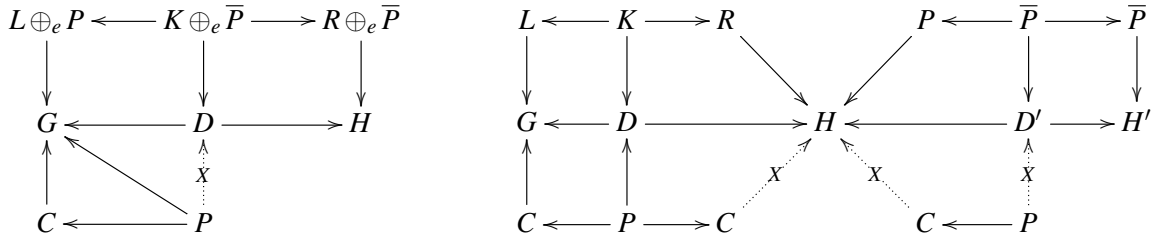

In both cases, $\overline{P}$ can be constructed as follows. Consider the set $S_{PatternP1,P}$ of morphisms identifying the maximal occurrence of $PatternP1$ from Figure~\ref{fig:dynamicConstraints} in a concrete graph $P$.  By restricting $PatternP1$ to the regions $mel$ and $els$, define the new pattern $PEls$ and the induced pattern morphism ('inclusion')  $m_{els}:PEls\rightarrow{PatternP1}$, as indicated in Figure~\ref{fig:constructionOfRepair}. For each morphism $f_i$ in $S_{PatternP1,P}$, call $\overline{P_i}$
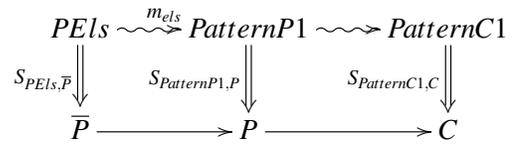
\begin{wrapfigure}{r}{0.42\textwidth}
\makebox{
\xymatrix{
        PEls \ar@{=>}[d]_{S_{PEls,\overline{P}}} \ar@{~>}[r]^{m_{els}} & PatternP1 \ar@{~>}[r] \ar@{=>}[d]_{S_{PatternP1,P}} &
				PatternC1 \ar@{=>}[d]_{S_{PatternC1,C}} \\
        \overline{P} \ar[r] & P \ar[r] & C
    }
}
\caption{The construction of $\overline{P}$.}
\label{fig:constructionOfRepair}
\end{wrapfigure}
the maximal image under $f_i$ of $PEls$ in $P$ and $S_{PEls,\overline{P_i}}$ the induced ('vertical') morphism.
Now the desired graph $\overline{P}$ is obtained as the colimit object\footnote{The use of $\overline{P}$ in the constructions of Figure~\ref{fig:repairs} is conservative,
as one could disrupt the premise by selecting any of the $\overline{P_i}$.} of the diagram defined by all the $S_{PEls,\overline{P_i}}$.
%
%
It is to be noted that several constraints of the first form can require that the element be typed in a certain way.
When adopting the first solution, one should find the colimit of all the premises, with all the elements identified. When adopting the second solution, the pattern in each premise can be disrupted individually, and we would have a collection of repair (disrupting) rules which can be applied in any order after applying the rule changing the type annotation. 

For constraints of the second form, we have two cases. The first is when the element whose type is changed is the only one which makes the constraint satisfied. In this case, the premise can be used as a negative application condition, so that the change is not allowed, or a new element of the correct type must be created to maintain the connection with the pattern. In the second case the application of a type-changing rule creates an instance of the pattern in the premise, thus requiring that an element of the correct type be connected to the pattern. Again, one can create such an element via a repair action. Alternately, if there already exists an element which is in the connection to the pattern required by the conclusion, a repair action can add a type annotation to this element.
Since there are no negative constraints and there is no limit to the number of type annotations for an element (except that it cannot be annotated twice with the same type), the second solution is always feasible.

Constraints of the third form can be enforced with techniques analogous to those for the first form, either including the pattern in the $R$ part of the rule, or adding it through a repair action after rule application.
The problem whether cascading repair actions of the second type can lead to infinite expansions of the graph remains to be solved.

\section{Case studies}\label{sec:caseStudies}
We present a number of case studies showing how the use of type annotation can model some typical situations from real world policies or from software modeling. In the figures illustrating them, we adopt the following conventions, in order to reduce cluttering:

\begin{enumerate}
\item Model instances are presented via a UML-like notation, where the domain to which they belong is represented as a type name.
\item Unique properties
%
%
are represented as attributes or as literals from some primitive domain.
\item The types
%
%
used for annotating,
i.e. with
%
%
\texttt{isMEType=true},
are represented as instances of \texttt{Type}, the specific sort being clear from the context, with an indication of the value of the attribute \texttt{name}.
\item Nodes for which \texttt{isTypeAnnotation=true} are represented as instances of \texttt{TypeAnnotation}.
\end{enumerate}

\subsection{Gender change}\label{sec:sex}

Consider the constraint \texttt{DriverIsMale} of Figure~\ref{fig:driverIsMale} describing one aspect of the driving law in Saudi Arabia, which prescribes that only men can have a driving licence. We are therefore in the \texttt{Person} domain, where elements can be typed according to their gender: \texttt{Male} or \texttt{Female}. The constraint is
\begin{wrapfigure}{r}{0.42\textwidth}
	\includegraphics[width=6cm]{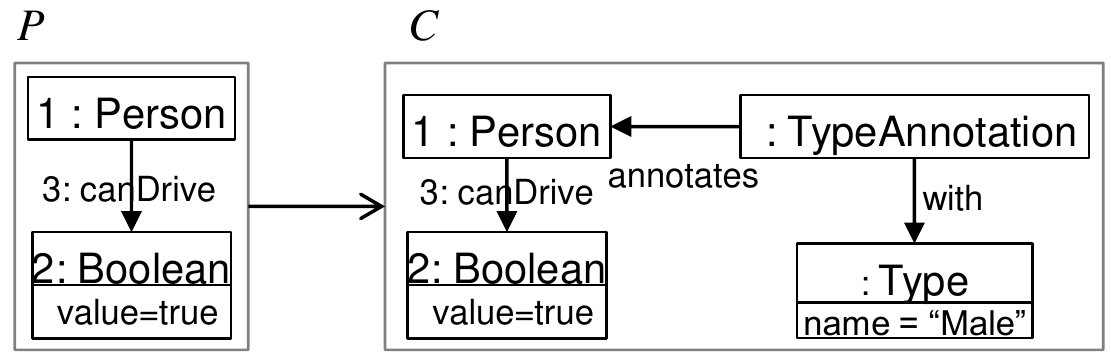}
	\caption{Constraint \texttt{DriverIsMale}}
%
	\label{fig:driverIsMale}
\end{wrapfigure}
 commonly overlooked for foreign nationals, but is, in principle, not waived.
Figure~\ref{fig:fromMaleToFemale} (left) describes the rule \texttt{FromMaleToFemale} modeling one direction of gender change.
A rule of this type is only possible with type annotation (or using gender as an attribute rather than a type), as the corresponding rule under a traditional typing morphism, where a person would actually be represented as an instance of either \texttt{Male} or \texttt{Female}, would require creating a new instance of \texttt{Female} and deleting the old instance of \texttt{Male}.

\begin{figure}[htb]
	\centering
	\subfigure{\includegraphics[width=10cm]{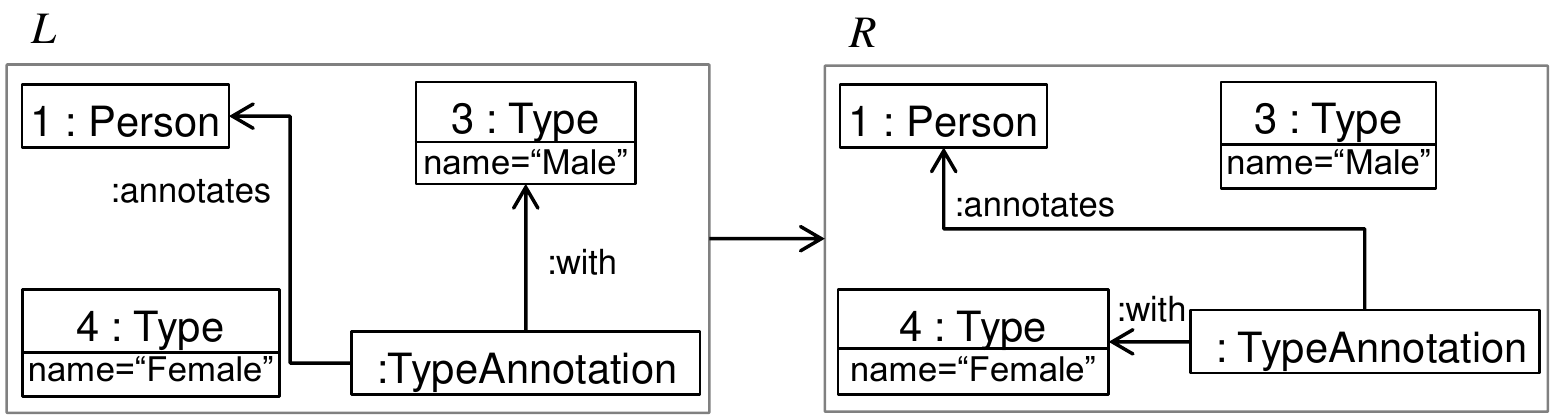}}
	\vspace{0.2cm}
	\vline
	\vspace{0.2cm}
	\subfigure{\includegraphics[height=2.75cm]{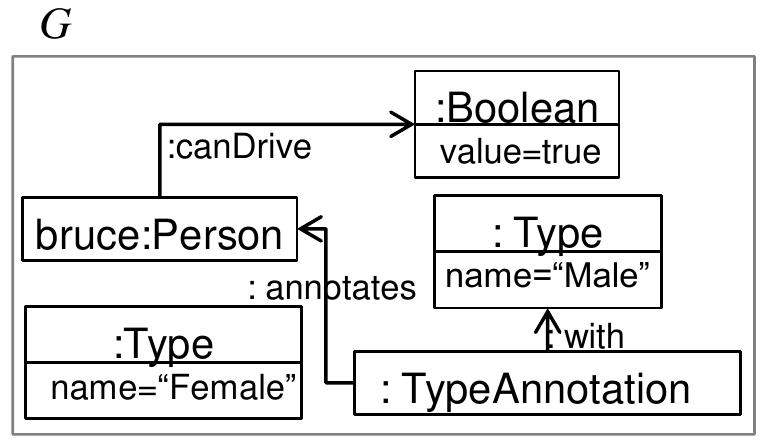}}
	\caption{Left: rule \texttt{FromMaleToFemale} changing one's gender. Right: a situation where gender change would lead to violate constraint \texttt{DriverIsMale}.}
	\label{fig:fromMaleToFemale}
\end{figure}

Now, if rule \texttt{FromMaleToFemale} is applied to the graph of
Figure~\ref{fig:fromMaleToFemale} (right), changing Bruce's gender, a situation violating
\texttt{DriverIsMale} is produced. To preserve correctness, proceeding according to one of the constructions from Figures~\ref{fig:repair1} or~\ref{fig:repairs} would remove the \texttt{canDrive} edge, either concurrently or after gender change, as in this case $\overline{P}$ is simply composed of the nodes \texttt{Person} and \texttt{Boolean}. Removing the edge after rule application would probably be more appropriate in modeling this situation.

\subsection{Classifications}\label{sec:classification}

Until 2006, the term \emph{minor planet} was adopted by the International Astronomical Union to indicate an astronomical object in direct orbit around the Sun that was neither a planet nor exclusively classified as a comet.
Whether an object is a planet, a minor planet or a comet depended on a number of measurable properties of that object.
After 2006, new classification criteria have been introduced and the terms \emph{dwarf planet} and
\emph{small Solar System body} (SSSB) have been used to sub-categorise minor planets. In particular, a planet is an element in the domain of
astronomical objects which:
\begin{inparaenum}[(1)]
\item is in orbit around the Sun,
\item has sufficient mass to assume hydrostatic equilibrium (a nearly round shape), and
\item has ``cleared the neighbourhood'' around its orbit.
\end{inparaenum}
The terms dwarf planet and SSSB are used for objects which fail the criteria for being a planet in specific ways. Figure~\ref{fig:astronomy} summarises the criteria for classifying an astronomical object as a planet, a dwarf planet or an SSSB in terms of our approach.

\vspace{-0.25cm}

\begin{figure}[htb]
\centering
\subfigure{\includegraphics[height=3.25cm]{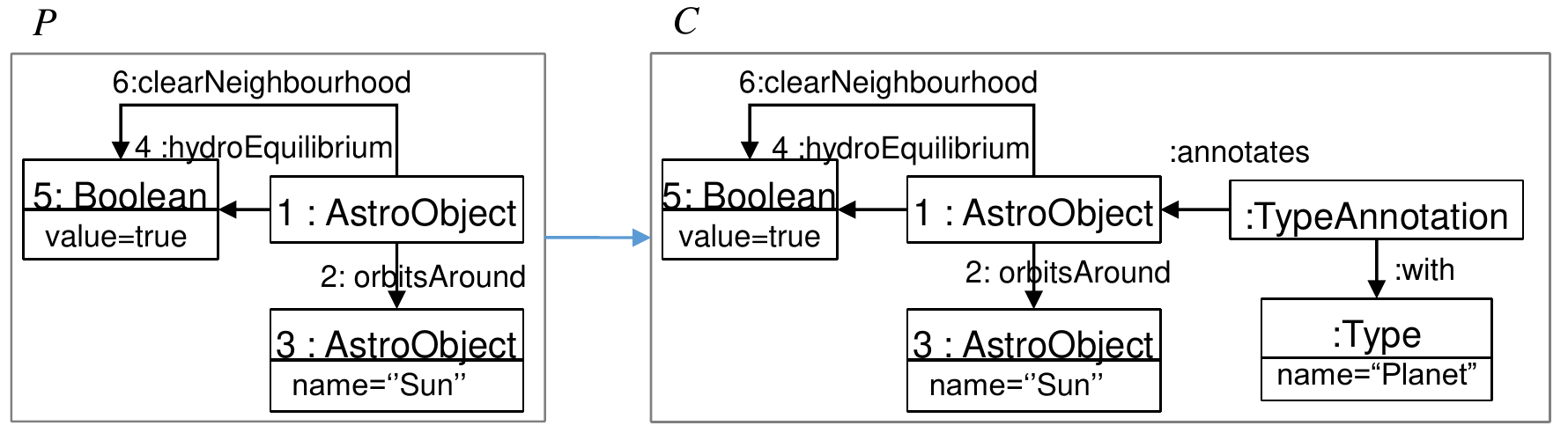}}
\subfigure{\includegraphics[height=3.5cm]{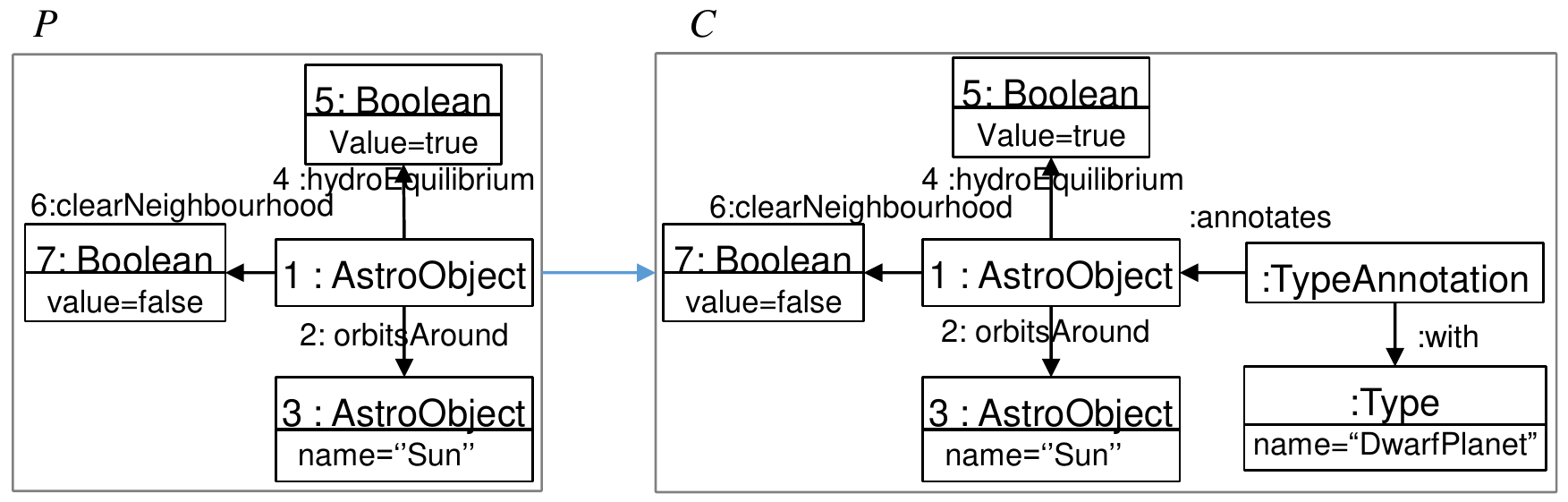}}
\subfigure{\includegraphics[height=3.5cm]{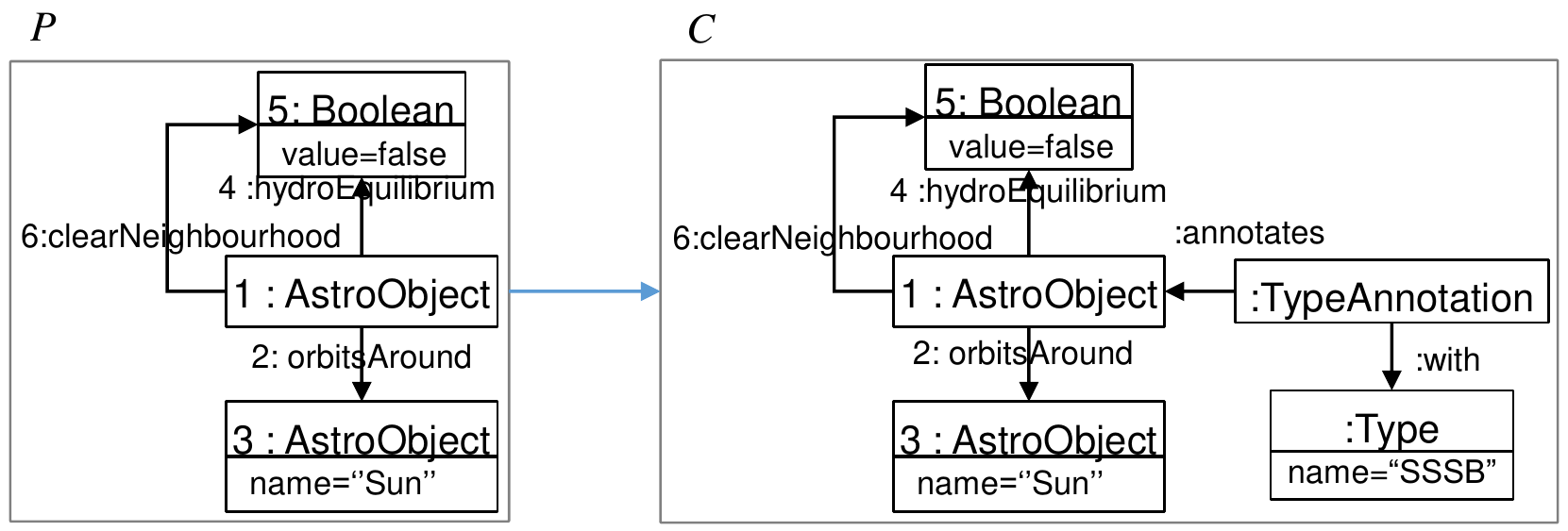}}
\caption{Constraints \texttt{isPlanet} (top), \texttt{isDwarfPlanet} (middle), and \texttt{isSSSB} (bottom) to classify an astronomical object as a planet, a dwarf planet, or an SSSB, respectively.}
	\label{fig:astronomy}
\end{figure}

As a consequence, a number of objects had to be reclassified into one of the categories. In particular, Pluto was no longer considered a planet and was classified as a dwarf planet. Figure~\ref{fig:fromPlanetToDwarf} presents the rule \texttt{fromPlanetToDwarf}, used to reclassify objects like Pluto.
This, and in general any form of classification based on measurement of some properties, is a case in which constraints of the first form are in use. A change of type can thus derive by the specification of the precise pattern of properties defining a type, or by the ability to perform more precise measures. It is also to be noted that for some bodies a dual classification as both minor planet and comet is admitted, pointing to the need for multiple typing.

\begin{figure}[htb]
\centering
	\includegraphics[height=3.5cm]{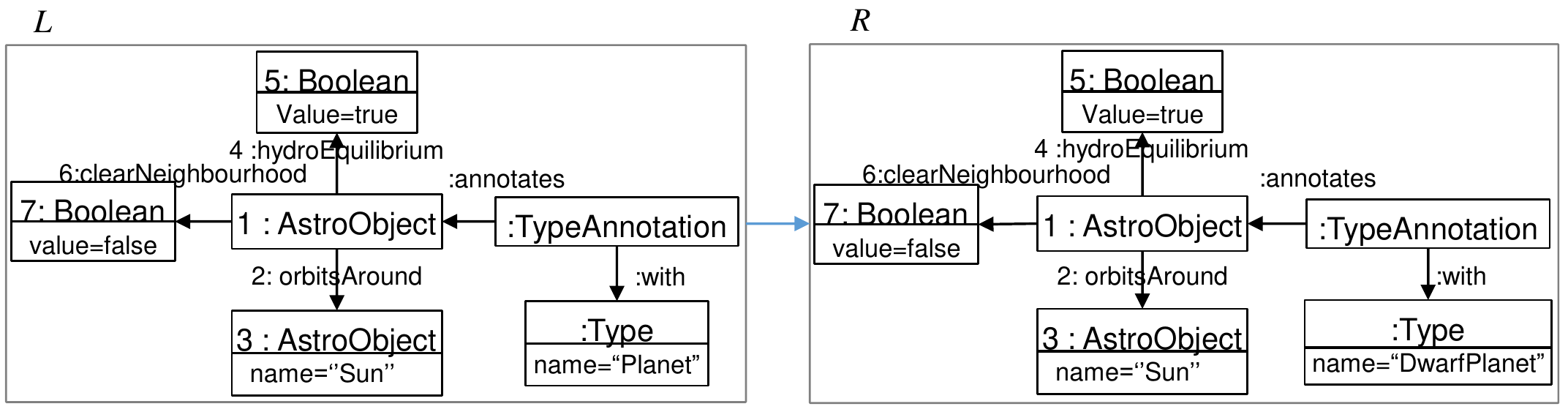}
	\caption{The rule \texttt{fromPlanetToDwarf} for reclassification of astronomical objects.}
	\label{fig:fromPlanetToDwarf}
\end{figure}

\subsection{Credentials}\label{sec:credentials}

Security credentials can be seen as a way to classify subjects into types, but it is not uncommon for individuals to belong to different types, without necessarily defining a supertype including all permissions of both types.  A manager of a particular governmental agency would have credentials necessary to access the resources necessary for the tasks assigned to her, and could also be a volunteer arbitrator in small claim litigations framework (within local county courts) with credentials to access the files of the cases assigned to her. There is no need to define a ``role'' to include both sets of functions to be able to type the subject uniquely. The manager decides, as a career move, to change governmental agency: the credentials needed for her former positions would have to be revoked, and new ones for her new position be reissued.  There is no need to revoke her arbitrator credentials (unless they have become incompatible with her new role). Hence part of the annotation remains unchanged while part is modified.

\subsection{Object-oriented programming and modeling}\label{sec:object}

In object-oriented programming, roles define different views of objects allowing the integration of different behaviours.
In contrast with~\cite{DBLP:journals/tois/GottlobSR96}, type annotations allow the definition of behaviours associated with roles without having to refer to a common root (apart from the top element in the type hierarchy).

Stereotypes in UML allow the addition of constraints on the instantiation of metaclasses, but they do not allow for the possibility of establishing specific relations between stereotyped elements, which is instead possible using constraints on elements with type annotation. Stereotypes are directly represented through type annotations, with the additional flexibility that we can bring the annotation at the level of instances and not only of metaclasses. 

\section{Conclusions}\label{sec:concls}
We have presented an approach to describe type information associated with elements of a graph in terms of annotations instead of morphisms. The resulting category of type-annotated graphs has a subcategory isomorphic to that of typed graphs, so that important properties of typing are preserved. We have also provided an interpretation of the relation between typed and type-annotated graphs via triple patterns, and described how inheritance can be managed in terms of annotations.
We have discussed how constraints can be defined to characterise properties required of type-annotated graphs and how to preserve type annotation correctness under transformations which change the type information for some elements. This allows the modeling of several situations where changes in the context require forms of dynamic typing to adapt to the new context.
We have considered only positive constraints, and the extension of the notion of typing with negative constraints, prescribing that elements of some type cannot be involved in some patterns, or that they cannot assume some specific values, will be the subject of future work.

In this line, one can also devise usages of typing annotations to deal with exceptions, as is often needed in ontologies and taxonomies. For example, human beings are uniquely characterised among primates by having 23 pairs of chromosomes. However, people with Down syndrome have an extra copy of chromosome 21, while women with Turner syndrome lack one X chromosome. In this sense, structural compliance with the human karyotype is a sufficient, but not necessary condition for being classified as human. As discussed in the paper, a type annotation can be associated with an element of the domain to indicate conformance with a certain pattern of observations, or to restrict the possibility of further specialisations. The management of exceptions, e.g. to relax structural constraints, would entail the definition of negative or nested constraints on type annotations, and is to be studied in future work.

\nocite{*}
\bibliographystyle{eptcs}
\bibliography{annotation}
\end{document}